\documentclass[journal]{IEEEtran}
\usepackage{blindtext, graphicx}
\usepackage{amssymb}
\usepackage{amsmath}
\usepackage{graphicx}
\usepackage{cite}
\usepackage{citesort}
\usepackage{multicol}
\usepackage{amsfonts}
\usepackage{color}
\usepackage{subfigure}
\usepackage{graphicx,epstopdf}
\usepackage{epsfig}	
\usepackage{cite,graphicx,amsmath,amssymb}
\usepackage{comment}
\usepackage{amsmath}
\usepackage{lipsum}
\usepackage{stfloats}
\newtheorem{proposition}{Proposition}

\newlength{\halfpagewidth}
\divide\halfpagewidth by 2

\newtheorem{theorem}{\textbf{Theorem}}

\newtheorem{lemma}{\textbf{Lemma}}

\newtheorem{corollary}{\textbf{Corollary}}

\newtheorem{proof}{\textbf{Proof}}

\makeatletter
\def\ScaleIfNeeded{%
\ifdim\Gin@nat@width>\linewidth \linewidth \else \Gin@nat@width
\fi } \makeatother

\begin{document}
%

\title{{Secure Communications in Millimeter Wave\\ Ad Hoc Networks}}

\author{
}
\author{Yongxu Zhu, Lifeng Wang,~\IEEEmembership{Member,~IEEE,}  Kai-Kit Wong,~\IEEEmembership{Fellow,~IEEE,} and Robert W. Heath, Jr.,~\IEEEmembership{Fellow,~IEEE}
\thanks{ Yongxu Zhu, L. Wang, and K.-K. Wong are with the Department of Electronic and Electrical Engineering, University College London, London, UK (Email: $\rm\{yongxu.zhu.13,lifeng.wang, kai$-$\rm kit.wong\}@ucl.ac.uk$).}
\thanks{Robert W. Heath, Jr. is with  Department of Electrical and Computer Engineering, The University of Texas at Austin, Texas, USA (Email: $\rm{rheath}@ece.utexas.edu$).}
}

\maketitle

\begin{abstract}
Wireless networks with directional antennas, like millimeter wave (mmWave) networks, have enhanced security.  {For} a large-scale mmWave ad hoc network in which eavesdroppers are randomly located, however, eavesdroppers can still intercept the confidential messages, since they may reside in the signal beam. {This} paper explores the potential of physical layer security in  mmWave ad hoc networks. Specifically, {{we characterize  the impact of mmWave channel characteristics, random blockages, and antenna gains on the secrecy performance.}}  For the special case of uniform linear array (ULA), a tractable approach is proposed to evaluate the average achievable secrecy rate. We also characterize the impact of artificial noise in such networks.  Our results reveal that in the low transmit power regime, the use of low mmWave frequency achieves better secrecy performance, and when increasing transmit power, a transition from low mmWave frequency to high mmWave frequency is demanded for obtaining a higher secrecy rate. More antennas at the transmitting nodes are needed to decrease the antenna gain obtained by the eavesdroppers when using ULA.  Eavesdroppers can intercept more information by using a wide beam pattern. Furthermore, the use of artificial noise may be ineffective for enhancing the secrecy rate.
\end{abstract}


\section{Introduction}
Wireless ad hoc networks have been widely applied in several areas including tactical
networks, device-to-device, and personal area networking. {{Unfortunately, interference from nearby transmitters severely deteriorate the throughput of ad hoc networks either through reducing the link quality, or reducing the number of links that can operate simultaneously.}} Due to the lack of central coordination, beamforming or directional antennas are one approach for suppressing interference~\cite{Kaibin_adhoc_2012}. Recently, millimeter wave (mmWave) has been viewed as a promising technology for supporting high-speed data rate in the mobile cellular systems~\cite{ZP2011}. MmWave with directional transmissions and large bandwidths provides rich opportunities for ad hoc networks. Compared to the lower frequency counterpart, mmWave ad hoc networks experience less interference and achieve greater rate coverage~\cite{Andrew_Thornburg_2014}.

Security in ad hoc networks is important~\cite{Lidong_zhou_1999}. The traditional higher-layer key distribution and management may increase the burden of transmitting confidential messages in such decentralized networks. Recent developments have shown that by leveraging the randomness inherent in wireless channels, physical layer security can be a low-complexity alternative for safeguarding complex wireless networks~\cite{huangjing2011}. {{By taking advantage of unique mmWave channel features, this paper establishes the  potential of physical layer security in mmWave ad hoc networks.}}

\subsection{Related Works and Motivation}
Early work has studied the effects of channel fading on physical layer security, see, e.g.,
~\cite{YingbinLiang2008,lifeng2014physical} and the references therein. The implementation of cooperative
jamming and artificial noise can degrade the eavesdropper's channel and further improve secrecy~\cite{huangjing2011,Zhou2010}. Recently, new network architectures and emerging transmission technologies such as  heterogeneous networks (HetNets) and massive multiple-input
multiple-output (MIMO) have promoted more research on physical layer security. In HetNets, dense small cells are deployed, which results in ubiquitous inter-tier and intra-tier interference. For secrecy communications at the physical layer, such interference
can be utilized for confounding the eavesdroppers. In~\cite{Tiejun_Lv}, spectrum allocation and transmit beamforming were designed
for maximizing the secrecy rate in a two-tier HetNet. {In \cite{H_M_Wang_2016}, an access threshold-based secrecy mobile association policy was proposed in a $K$-tier HetNet.} Massive MIMO uses large number of antennas to provide high array gains for legitimate receivers. The work of~\cite{Jue_Wang_2015_MISO} studied the case of jamming when the transmitter equipped with large number of antennas served one single-antenna receiver. It was shown
in~\cite{junzhu2014} that the application of random artificial noise in massive MIMO cellular networks can achieve a better performance/complexity tradeoff compared to the
 conventional null space based artificial noise. In \cite{Wang_Dec_2016}, secrecy and energy efficiency in massive MIMO aided heterogeneous C-RAN was studied, which showed that the centralized and distributed large-scale antenna systems can coexist to enhance the secrecy and cut power consumption. {{While the aforementioned literature has provided a solid understanding of physical layer security in the wireless systems with lower-frequency bands (sub-6 GHz), the research on mmWave secrecy communication is in its infancy. }}


Physical layer security in decentralized wireless networks such as sensor and ad hoc type of networks has been investigated in ~\cite{Xiangyun_2011_Aug,Deng_WSN_2016,XiZhang_2013_TIFS,Cai_2014}. {{{In \cite{Xiangyun_2011_Aug}, secrecy transmission capacity under connection outage and secrecy outage concerns was examined in an ad hoc network, in which both legitimate nodes and eavesdroppers are randomly distributed.}}} In \cite{Deng_WSN_2016},  the average achievable secrecy rate was examined in a three-tier sensor networks consisting of sensors, access points and sinks, and it was shown that there exists optimal number of access points for maximizing the average achievable secrecy rate. Secrecy enhancement in ad hoc networks was studied in \cite{XiZhang_2013_TIFS}, where two schemes for the generation of artificial noise were compared. In \cite{Cai_2014}, relay transmission in ad hoc networks was evaluated from the perspective of security connectivity. {{Again, these works solely focus on the lower-frequency secrecy communications in decentralized wireless networks.}}

Due to the peculiar mmWave channel characteristics, physical layer security in mmWave systems has recently attracted much interest~\cite{Valliappan2013,L_Wang_2014,yang2015safeguarding,Steinmetzer_2015}. In  \cite{Valliappan2013},  mmWave antenna subset modulation was designed to secure point-to-point communication by introducing randomness in the received constellation, which confounds the eavesdropper.  In~\cite{L_Wang_2014}, the mmWave multiple-input, single-output, multiple-eavesdroppers channel was considered in a single cell, and {{it was indicated  that high-speed secure link at the mmWave frequencies could be reached with the assistance of large antenna arrays and large mmWave bandwidths.}} The work of \cite{yang2015safeguarding} illustrated the impacts of key factors such as large bandwidth and directionality on the physical layer security in mmWave networks, and provided more opportunities and challenges in this field. {In \cite{Steinmetzer_2015}, it was shown that} even only one eavesdropper may be able to successfully intercept highly directional mmWave transmission. In the work of \cite{Steinmetzer_2015}, although the eavesdropper was located outside the signal beam, reflections could be exploited by the eavesdropper that used small-scale reflectors within the beam, which has little blockage effect on the legitimate receiver's performance.
{Secrecy outage of an mmWave cellular network was analyzed in \cite{chao_wang_2016}, where authorized users and eavesdroppers were  assumed to be  single-omnidirectional-antenna nodes. In~\cite{Tharmalingam_2016}, secrecy outage of a mmWave overlaid microwave network was derived by considering a specific blockage model and assuming that mmWave channel undergoes Nakagami-$m$ fading for tractability. In two-way amplify-and-forward MIMO relaying networks, \cite{Shiqing_Gong_2016} proposed mmWave secrecy beamforming schemes to maximize the secrecy sum rate.}

{{ Prior work only pays attention to the physical layer security in lower-frequency ad hoc networks. In mmWave ad hoc networks, the directional communication with narrow beam is more robust against eavesdropping. {{The mmWave link is sensitive to the blockage and experiences higher propagation loss, and mmWave channel undergoes rapid fluctuation and has much lower coherence time than the today's networks because of much larger Doppler spread~\cite{S_Rangan_2014}. Hence mmWave link is more random and hard to be intercepted by malicious eavesdroppers compared to the low-frequency counterpart.}}

\subsection{Approach and Contributions}
This paper studies physical layer security in  mmWave ad hoc networks. Our analysis accounts for the key features of mmWave channel and the effects of different antenna array gains and node densities. The detailed contributions and insights are summarized as follows.
\begin{itemize}
  \item We model the mmWave ad hoc networks with the help of stochastic geometry, to characterize the random spatial locations of  transmitting nodes and eavesdroppers. {{The effect of  blockage is also incorporated such that links are either line-of-sight (LoS) or  non-line-of-sight (NLoS).}} The average achievable secrecy rate is derived to quantify the impacts of key system parameters such as antenna gain, transmitting node and eavesdropper densities on the secrecy performance. Our results show that with increasing transmit power, a transition from low mmWave frequency to high mmWave frequency is needed for achieving better secrecy performance. {Compared} to eavesdropping, the performance is dominated by the surrounding interference in the high node density case. {{The use of different mmWave frequencies has a big impact on the secrecy performance, which needs to be carefully selected in practice.}}

  \item We develop an approach to evaluate the average achievable secrecy rate when utilizing uniform linear array (ULA). Our results show that {{adding more antennas at the transmitting node decreases antenna gains obtained by eavesdroppers.}}

  \item We examine the impact of artificial noise on the secrecy rate. Our results show that in mmWave ad hoc networks, the use of artificial noise can still enhance the secrecy when power allocation between the information signal and artificial noise is properly set. {{Moreover, the use of artificial noise may have an adverse effect on the secrecy rate in the low node density scenarios,  where more transmit power should be allocated to improve the transmission rate between the transmitting node and its intended receiver.}}
\end{itemize}

The remainder of this paper is organized as follows. Section II presents the network and the mmWave channel model. Section III evaluates the average achievable secrecy rate of this network and also discusses the implementation of uniform linear array. Section IV analyzes the use of artificial noise on the secrecy performance.  Numerical results are provided in Section V and conclusion is drawn in Section VI.

\section{System Description}
Consider a {mmWave} ad hoc network, where a group of transmitting nodes are randomly distributed following a homogeneous Poisson point process (PPP) $\Phi$ with $\lambda$. The dipole model is adopted~\cite{Baccelli2009}, where the distance for a typical transmitting node-receiver is fixed at $r$, and the typical receiver is assumed to be located at the origin. Both the transmitting node and its corresponding receiver use  {directional} beamforming for data transmission, which is intercepted by multiple eavesdroppers. {We} consider the case of passive eavesdropping without any active attacks to deteriorate the information transmission. The locations of eavesdroppers are modeled following an independent homogeneous PPP $\Phi_e$ with  $\lambda_e$. {{We consider the directional beamforming and use a sectored model to analyze the beam pattern~\cite{Andrew_Thornburg_2014,A_M_Hunter_2008,Tianyang_arxiv2014,Singh_2015} (See Fig. 1 in~\cite{Andrew_Thornburg_2014})}}, i.e.,  the effective antenna gain for an interferer $i$ seen by the typical receiver is expressed as
\begin{align}\label{array_gain_pattern}
{G_i} = \left\{ \begin{array}{l}
{G_\mathrm{M}^2}{\rm{,}}~\qquad {\Pr _{\mathrm{MM}}}{\rm{ = }}{\left( {\frac{\theta }{2\pi }} \right)^2} ,\\
{G_\mathrm{M}}{G_\mathrm{m}}{\rm{,}}\quad {\Pr _{\mathrm{Mm}}}{\rm{ = }}\frac{{\theta \left( {2\pi  - \theta } \right)}}{{{(2\pi)^2}}} ,\\
{G_\mathrm{m}}{G_\mathrm{M}}{\rm{,}}\quad {\Pr _{\mathrm{Mm}}}{\rm{ = }}\frac{{\theta \left( {2\pi  - \theta } \right)}}{{{(2\pi)^2}}} ,\\
{G_\mathrm{m}^2},~\qquad {\Pr _{\mathrm{mm}}} = {\left( {\frac{{2\pi  - \theta }}{2\pi }} \right)^2} ,
\end{array} \right.
\end{align}
where $G_{\mathrm{M}}$ denotes the main-lobe gain with the beamwidth $\theta$,  $G_{\mathrm{m}}$ denotes the back-lobe gain, and $\mathrm{Pr}_{\ell k}$ ($\ell, k \in \left\{\mathrm{M},\mathrm{m}\right\}$) denotes the probability that the antenna gain $G_\ell G_k$ occurs. We assume that the maximum array gain ${G_\mathrm{M}}{G_{\mathrm{M}}}$ is obtained for the typical transmitting node-receiver.

In light of the blockage effects in the outdoor scenario, the signal path can be LoS or  NLoS.  We denote $f_\mathrm{Pr}\left(R\right)$ as the probability that a link at a distance $R$ is LoS, while the NLoS probability of a link is $1-f_\mathrm{Pr}\left(R\right)$. The LoS probability function $f_\mathrm{Pr}\left(R\right)$ can be obtained from
 field measurements or stochastic blockage models~\cite{Tianyang_arxiv2014}.


 We employ a short-range propagation model in which given a distance $\left|X_i\right|$, the path loss function is denoted as $L\left(\left|X\right|\right)=\beta { {{\left( \max\left(d,\left|X\right|\right) \right) }}^{ - {\alpha}}}$ with a reference distance $d$~\cite{Baccelli2006_gen}, where $\alpha$ is the path loss exponent depending on the LoS or NLoS link, namely $\alpha=\alpha_\mathrm{LoS}$ for LoS link and $\alpha=\alpha_\mathrm{NLoS}$ for NLoS link, and {$\beta$ is the frequency independent constant parameter of the path loss, which is commonly set as ${(\frac{{\text{c}}}{{4\pi {f_c}}})^2}$ with $c=3 \times 10^8 \rm m/s$ and the carrier frequency $f_c$. Hence there are different $\beta$ values for different mmWave frequencies, which allows us to examine the effects of using different mmWave frequencies.}  Note that the sparse scattering mmWave environment makes many traditional fading distributions invalid for the modeling of the mmWave channel~\cite{Ayach2014}. For tractability, we neglect small scale fading as~\cite{TED2013IEEE_Access} argues that fading is not significant in LOS links with significant beamforming.   Hence the signal-to-interference-plus-noise ratio (SINR) at a typical receiver is written as
 \begin{align}\label{SINR_typical_receiver}
{\gamma _o} = \frac{{{P_t}G_{\text{M}}^2 L\left(r\right)   }}{{\sum\nolimits_{i \in \Phi/o} {{P_t}{G_i} L\left(\left|X_i\right|\right)    }  + {\sigma_o^2}}},
 \end{align}
 where $P_t$ denotes the transmit power,   $\left| {{X_{i}}} \right|$ is the distance between the typical receiver and the interferer $i \in \Phi/o$~(except the typical transmitting node),  and $\sigma_o^2$ is the noise power.

When the eavesdropping channel is degraded under the effect of interference, secrecy indeed becomes better. In this paper, we focus on the worst-case eavesdropping scenario, where all the eavesdroppers can  {mitigate} the interference. In fact, eavesdroppers are usually assumed to have strong ability, and they may cooperate with each other to cancel the interference, as seen in~\cite{Geraci_downlink}. We assume that the eavesdropping channels are independent of the legitimate channel\footnote{{We highlight that the secrecy in the mmWave correlated wiretap channel is a novel and important research area, and the existing contributions at lower frequencies can be seen in~\cite{Jeon_H_2011}.}}. In such a scenario, the most malicious eavesdropper that has the largest SINR of the received signal dominates the secrecy rate~\cite{LunDong}. Thus, the SINR at the most malicious eavesdropper is written as
 \begin{align}\label{SINR_eavesdropping}
{\gamma _{{e^*}}} = \mathop {\max }\limits_{e \in {\Phi _e}} \left\{ {\frac{{{P_t}{G_e} L\left(\left|X_e\right|\right)     }}{{{\sigma_e^2}}}} \right\},
 \end{align}
where $\left|X_e\right|$ is the distance between the typical transmitting node and the eavesdropper $e \in {\Phi _e}$, $\sigma_e^2$ is the power of noise and weak interference, and $G_e$ is the antenna gain seen from the eavesdropper $e \in {\Phi _e}$  described by
\begin{align}\label{array_gain_pattern_eavesdropper}
{G_e} = \left\{ \begin{array}{l}
{G_\mathrm{M}}{G^e_{\mathrm{M}}}{\rm{,}}\quad {\Pr _{\mathrm{MM}}}{\rm{ = }}{ {\frac{\theta \phi}{\left(2\pi\right)^2}} } ,\\
{G_\mathrm{M}}{G^e_\mathrm{m}}{\rm{,}}\quad {\Pr _{\mathrm{Mm}}}{\rm{ = }}\frac{{\theta \left( {2\pi  - \phi } \right)}}{{\left(2\pi\right)^2}} ,\\
{G_\mathrm{m}}{G^e_\mathrm{M}}{\rm{,}}\quad {\Pr _{\mathrm{Mm}}}{\rm{ = }}\frac{{ \left({2\pi  - \theta }\right)\phi }}{{{\left(2\pi\right)^2}}} ,\\
{G_\mathrm{m}}{G^e_\mathrm{m}},\quad {\Pr _{\mathrm{mm}}} = \frac{{(2\pi  - \theta)(2\pi  - \phi) }}{\left({2\pi }\right)^2} ,
\end{array} \right.
\end{align}
in which $\phi$, ${G^e_{\mathrm{M}}}$ and ${G^e_\mathrm{m}}$ are the beamwidth of the main-lobe, main-lobe gain and back-lobe gain of the beam pattern used by the eavesdropper $e \in {\Phi _e}$, respectively.

\section{Secrecy Evaluation}
In this section, we analyze the average achievable secrecy rate in mmWave ad hoc networks. As shown in~\cite{Bloch}, physical layer security is commonly characterized by the secrecy rate $R_\mathrm{s}$, which is defined as
\begin{align}\label{secrecy_rate}
{R_\mathrm{s}} = {\left[{{{\log }_2}\left( {1 + {\gamma_o}} \right) -{{\log }_2}\left( {1 + {\gamma_{e^{*}}}} \right)} \right]^+}.
\end{align}

{  Based on \eqref{secrecy_rate}, we have the following proposition.
\begin{proposition}\label{Proposition_1}
In mmWave ad hoc networks, the average achievable secrecy rate  is given by
\begin{align}\label{average_Secrecy_rate}
{{\overline{R}}_\mathrm{s}} = \left[{\overline{R}}-{\overline{R}}_{e^{*}}\right]^+,
\end{align}
where $ [x]^+=\max\{x,0\} $, ${\overline{R}}=\mathbb{E}\left[{{\log }_2}\left( {1 + {\gamma_o}} \right)\right]$ is the average rate of the channel between the typical transmitting node and its receiver, and ${\overline{R}}_{e^{*}}=\mathbb{E}\left[ {{\log }_2}\left( {1 + {\gamma_{e^{*}}}} \right) \right]$ is the average rate of the channel between the typical transmitting node and the most malicious eavesdropper.
\end{proposition}}
{ \begin{proof}
We first show that the average rate ${\overline{R}}$ is achievable by considering the  fact that mmWave channel experiences rapid fluctuation, and  the coherence time in mmWave frequencies is around an order of magnitude lower than that at sub-6 GHz as the Doppler shift linearly scales with frequency~\cite{S_Rangan_2014,Popovski_TCOM_2016}. Moreover, mmWave links  undergo more dramatic swings between LoS and NLoS  due to the high level of shadowing~\cite{S_Rangan_2014}. Therefore, coding over many coherence intervals is possible, and thus the average rate ${\overline{R}}$ can be achieved.

Since the malicious eavesdroppers only intercept the secrecy massages passively without any transmissions, the channel state information (CSI) of the eavesdropping channels cannot be obtained by the transmitting node, and the transmission rate of a typical transmitting node is only dependent on the CSI of the channel between itself and the typical receiver. In addition, the maximum average rate in an arbitrary wiretap channel cannot exceed ${\overline{R}}_{e^{*}}$. As such, we obtain the average achievable secrecy rate in mmWave ad hoc networks as \eqref{average_Secrecy_rate}.
\end{proof}}

To evaluate the average achievable secrecy rate, we first derive the average rate ${\overline{R}}$, which is given by the following theorem.
\begin{theorem}\label{theorem_1}
The exact average rate between the typical transmitting node and its intended receiver  is given by
\begin{align}\label{R_average_rate}
{\overline{R}}=\frac{1}{\ln2}\int_0^\infty  {\frac{1}{{z}}(1 - \Xi_1(z))\Xi_2(z) {e^{ -z\sigma_o^2}}dz},
\end{align}
where $\Xi_1(z)$ and $\Xi_2(z) $ are respectively given by \eqref{Xi_z_x} and \eqref{Xi_z_1} at the top of next page{\color{blue}{\footnote{{We consider that the typical legitimate channel and the interfering channels are independent, due to the fact that the coherence time of mmWave channel is around an order of hundreds of microseconds and much shorter than today's cellular systems, and mmWave links experience more dramatic swings in path loss~\cite{S_Rangan_2014}.}}}}.
\begin{figure*}[!t]
\normalsize
\begin{align}\label{Xi_z_x}
\Xi_1(z)=f_\mathrm{Pr}\left(r\right)e^{- z{{P_t}G_{\text{M}}^2\beta \left( {\max {{\{ r,d\} }}} \right)^{ - {\alpha_\mathrm{LoS}}}}}
+(1-f_\mathrm{Pr}\left(r\right))e^{- z{{P_t}G_{\text{M}}^2\beta \left( {\max {{\{ r,d\} }}} \right)^{ - {\alpha_\mathrm{NLoS}}}}}
\end{align}
\hrulefill 
\begin{align}\label{Xi_z_1}
\Xi_2(z)=\exp\Big( - 2\pi \lambda \int_0^\infty  {f_{{\rm{Pr}}}}\left( u \right) ( 1 -\Omega_1(z,u)) udu-2\pi \lambda \int_0^\infty (1-  {f_{{\rm{Pr}}}}\left( u \right)) ( 1 -\Omega_2(z,u)) udu\Big)
\end{align}
with
\begin{equation*} 
\left\{\begin{aligned}
\Omega_1(z,u)&=\sum\limits_{\ell, k  \in \left\{ {{\text{M}},{\text{m}}} \right\}} {\Pr }_{\ell k } \times
{{{e}}^{ - z{P_t}{G_\ell G_k}\beta {\left( {\max {{\{ u,d\} }}} \right)^{ - {\alpha_{{\rm{LoS}}}}}}}}\\
\Omega_2(z,u)&=\sum\limits_{\ell, k  \in \left\{ {{\text{M}},{\text{m}}} \right\}} {\Pr }_{\ell k } \times
{{{e}}^{ - z{P_t}{G_\ell G_k}\beta {\left( {\max {{\{ u,d\} }}} \right)^{ - {\alpha_{{\rm{NLoS}}}}}}}}
\end{aligned}\right.
\end{equation*}
\hrulefill
\vspace*{0pt}
\end{figure*}
\end{theorem}

\begin{proof}
See Appendix A.
\end{proof}
The exact average rate given in \eqref{R_average_rate} can be lower bounded as a simple expression, which is as follows.
\begin{theorem}\label{theorem_1_1}
The lower bound of the average rate ${\overline{R}}$ is given by
\begin{align}\label{LB_average_rate_Alice}
{\overline{R}}^{\mathrm{L}}_1= {\log _2}\left( 1 + \frac{{G_{\text{M}}^2\beta }r^{-{\overline{\alpha}}}}{{ {  \lambda \bar{ G} \Lambda  + \frac{N_o}{P_t}} }}  \right),
\end{align}
where ${\overline{\alpha}}=\left({\alpha _{{\text{LoS}}}} - {\alpha _{{\text{NLoS}}}}\right)f_\mathrm{Pr}\left(r\right) + {\alpha _{{\text{NLoS}}}}$, the average antenna gain ${{\bar G} }= \sum\nolimits_{\ell ,k \in \left\{ {{\text{M}},{\text{m}}} \right\}} {{G_{\ell }G_k}\text{Pr}_{\ell  k}}$, and $\Lambda$ is
\begin{align}\label{Lambda_1_1_1}
&\Lambda =\beta 2\pi  \Big( \int_{\text{0}}^d \big({({d^{ - {\alpha _{{\text{LoS}}}}}} - {d^{ - {\alpha _{{\text{NLoS}}}}}})} r{f_{{\text{Pr}}}}\left( r \right) + {d^{ - {\alpha _{{\text{NLoS}}}}}}r \big) dr \nonumber \\
& + \int_d^\infty  \big({({r^{1 - {\alpha _{{\text{LoS}}}}}} - {r^{1 - {\alpha _{{\text{NLoS}}}}}})}
  {f_{{\text{Pr}}}}\left( r \right) + {r^{1 - {\alpha _{{\text{NLoS}}}}}} \big) dr\Big).
\end{align}
 When the LoS probability is $f_\mathrm{Pr}\left(R\right)=e^{-\varrho R}$~\cite{Tianyang_arxiv2014}, \eqref{LB_average_rate_Alice} reduces to a closed-form expression with
\begin{align}\label{Lambda_r}
\Lambda  &= \beta 2\pi   \times\nonumber \\
&\Big[ \frac{{1 - {e^{ - {{d}}\varrho }}(1 + {{d}}\varrho )}}{{{\varrho ^2}}}(\frac{1}{{{d^{{\alpha _{{\text{LoS}}}}}}}} - \frac{1}{{{d^{{\alpha _{{\text{NLoS}}}}}}}}) + \frac{{\Gamma (2 - {\alpha _{{\text{LoS}}}},d\varrho )}}{{{\varrho ^{2 - {\alpha _{{\text{LoS}}}}}}}} \nonumber \\
&+ \frac{{{\alpha _{{\text{NLoS}}}} \cdot {d^{2 - {\alpha _{{\text{NLoS}}}}}}}}{{2({\alpha _{{\text{NLoS}}}} - 2)}} - \frac{{\Gamma (2 - {\alpha _{{\text{NLoS}}}},d\varrho )}}{{{\varrho ^{2 - {\alpha _{{\text{NLoS}}}}}}}}\Big].
\end{align}
\end{theorem}
\begin{proof}
See Appendix B.
\end{proof}
 From \textbf{Theorem 2}, we find that as the transmit power grows large, the average rate is asymptotically lower bounded as  ${\overline{R}}^{\mathrm{L}}_1\rightarrow {\log _2}\left( 1 + \frac{{G_{\text{M}}^2\beta }r^{-{\overline{\alpha}}}}{{ { \lambda \bar{G} \Lambda} }}  \right)$. It is explicitly shown from \eqref{LB_average_rate_Alice} that the average rate between the typical transmitting node and its receiver is a decreasing function of transmitting node density, and increases with narrower beam due to the lower average interfering antenna gain. In addition, we have the following important corollary.
\begin{corollary}
Given a required average rate $\overline{R}_{\mathrm{th}}$ between the typical transmitting node and its receiver, it is achievable when
the transmitting node density in the mmWave ad hoc network satisfies
\begin{align}\label{corollary_1_1}
\lambda \leq \left(\frac{{G_{\text{M}}^2\beta }r^{-{\overline{\alpha}}}}{2^{\overline{R}_{\mathrm{th}}}-1}-\frac{N_o}{P_t}\right) {\bar G}^{-1} \Lambda^{-1}.
\end{align}
From \eqref{corollary_1_1}, we see that narrower beams allow mmWave ad hoc networks to accommodate more transmitting nodes.
\end{corollary}

We next derive the average rate between the typical transmitting node and the most malicious eavesdropper, which is given by the following theorem.
\begin{theorem}\label{theorem_2}
The exact average rate between the typical transmitting node and the most malicious eavesdropper is given by
\begin{align}\label{Eve_average_rate}
{\overline{R}}_{e^{*}}= \frac{1}{{\ln 2}}\int_0^\infty  {\frac{{\left( {1 - \mathcal{P}_1\left( {x} \right)\mathcal{P}_2\left( {x} \right)} \right)}}{{1 + x}}dx},
\end{align}
where $\mathcal{P}_1\left( {x} \right)$ and $\mathcal{P}_2\left( {x} \right)$ are given  {in \eqref{PP_1} and \eqref{PP_2} with    ${\rm{\mathbf{1}}}\left(A\right)$ representing the indicator function that returns one if the condition $A$ is satisfied.}
\begin{figure*}[!t]
\normalsize
\begin{align}\label{PP_1}
\mathcal{P}_1\left( {x} \right)=\exp \left\{   - 2\pi {\lambda _e} \int_0^\infty  {{f_{\Pr }}({r_e})} {r_e} \sum\limits_{
  \ell,n  \in \left\{ {{\text{M}},{\text{m}}} \right\} }  {\rm{\mathbf{1}}}\left(  {\max \{ {r_e},d\}} < \big(\frac{{P_t}{G_\ell }{G_n^e}\beta}{x \sigma_e^2}\big)^{\frac{1}{{\alpha _{\mathrm{LoS}}}}}\right) {{{\Pr }_{\ell n}}} d{r_e} \right\}
\end{align}
\hrulefill 
\begin{align}\label{PP_2}
\mathcal{P}_2\left( {x} \right)=\exp \left\{   - 2\pi {\lambda _e}  \int_0^\infty (1-{{f_{\Pr }}({r_e})}) {r_e} \sum\limits_{
  \ell,n  \in \left\{ {{\text{M}},{\text{m}}} \right\} }  {\rm{\mathbf{1}}}\left(  {\max \{ {r_e},d\}} < \big(\frac{{P_t}{G_\ell }{G_n^e}\beta}{x \sigma_e^2}\big)^{\frac{1}{{\alpha _{\mathrm{NLoS}}}}}\right) {{{\Pr }_{\ell n}}} d{r_e} \right\}
\end{align}
\hrulefill 
\end{figure*}
\end{theorem}

\begin{proof}
See Appendix C.
\end{proof}

Substituting \eqref{R_average_rate} and \eqref{Eve_average_rate} into \eqref{secrecy_rate}, we can thus evaluate the average achievable secrecy rate in this network.

{{\subsection{Simplified LoS MmWave Model}
The aforementioned analysis is derived by considering an arbitrary LoS probability, which is general. In this subsection, we employ a simplified LoS mmWave model, as mentioned in~\cite{Tianyang_arxiv2014,J_Park_2016}. In this model, the mmWave link is LoS if the distance for a typical transmitting node-receiver is not larger than the maximum LoS distance $\mathrm{D}_{\mathrm{LoS}}$, and otherwise it is outage. When an LoS link between a typical transmitting node and its receiver is established (i.e., $r<\mathrm{D}_{\mathrm{LoS}}$), the exact average rate between the typical transmitting node and its intended receiver given in \textbf{Theorem 1} can be simplified as
\begin{align}\label{LoS_Rate_121}
{\hat{R}}=\frac{1}{\ln2}\int_0^\infty  {\frac{1}{{z}}(1 -e^{- z{{P_t}G_{\text{M}}^2 L\left(r\right)}})\hat{\Xi}_2(z) {e^{ -z\sigma_o^2}}dz},
\end{align}
where $\hat{\Xi}_2(z)$ is calculated as
\begin{align}\label{hat_Xi_2}
&\hat{\Xi}_2(z)= \exp\Bigg\{ - 2\pi \lambda \bigg[\frac{\mathrm{D}_{\mathrm{LoS}}^2}{2}-\sum\limits_{\ell, k  \in \left\{ {{\text{M}},{\text{m}}} \right\}} {\Pr }_{\ell k }
\Big( \nonumber\\
& \frac{d^2}{2} e^{- z{{P_t}G_{\text{M}}^2\beta d^{ - {\alpha_\mathrm{LoS}}}}} + \alpha_\mathrm{LoS}^{-1}(z{P_t}{G_\ell G_k}\beta)^{2/\alpha_\mathrm{LoS}} \nonumber\\
&\Big(\Gamma\Big(-\frac{2}{\alpha_\mathrm{LoS}},z{P_t}{G_\ell G_k}\beta\mathrm{D}_{\mathrm{LoS}}^{-\alpha_\mathrm{LoS}}\Big) \nonumber\\
& -\Gamma\Big(-\frac{2}{\alpha_\mathrm{LoS}},z{P_t}{G_\ell G_k}\beta d^{-\alpha_\mathrm{LoS}}\Big) \Big)
\Big)\bigg]  \Bigg\}.
\end{align}
Here,  $\Gamma\left(\cdot,\cdot\right)$ is the upper incomplete gamma function~\cite[(8.350)]{gradshteyn}.

It is explicitly shown from \eqref{LoS_Rate_121} that ${\hat{R}}$ is a decreasing function of $\lambda$, since adding more transmitting nodes results in larger interference.

Likewise, the exact average rate between the typical transmitting node and the most malicious eavesdropper given in \textbf{Theorem 3} can be simplified as
\begin{align}\label{Eve_average_rate_LoS_mmWave}
{\overline{R}}_{e^{*}}= \frac{1}{{\ln 2}}\int_0^\infty  {\frac{1-\exp\left(-2\pi {\lambda _e} \hat{F}_e\left( {x} \right) \right)}{{1 + x}}dx},
\end{align}
where the cumulative distribution function is given by
\begin{align}\label{Eves_SINR_11}
\hat{F}_e\left( {x} \right)=\sum\limits_{
  \ell,n  \in \left\{ {{\text{M}},{\text{m}}} \right\} }\left( {\rm{\mathbf{1}}}\left( d < \eta\left(G_\ell,G_n^e,x\right)\right)
  \frac{d^2}{2} + \frac{\varrho^2-d^2}{2} \right){{{\Pr }_{\ell n}}}
\end{align}
with $\eta\left(G_\ell,G_n^e,x\right)=\big(\frac{{P_t}{G_\ell }{G_n^e}\beta}{x \sigma_e^2}\big)^{\frac{1}{{\alpha _{\mathrm{LoS}}}}}$ and $\varrho=\min\left(\mathrm{D}_{\mathrm{LoS}},\eta\left(G_\ell,G_n^e,x\right)\right)$.

It is explicitly shown from \eqref{Eve_average_rate_LoS_mmWave} that ${\overline{R}}_{e^{*}}$ is an increasing function of $\lambda _e$, which means that the exact average rate between the typical transmitting node and the most malicious eavesdropper increases with the number of eavesdroppers.

Substituting \eqref{LoS_Rate_121} and \eqref{Eve_average_rate_LoS_mmWave} into \eqref{average_Secrecy_rate}, we can obtain the average achievable secrecy rate.
}}
\subsection{Uniform Linear Array}
We proceed to evaluate the secrecy performance when all the nodes in this networks are equipped with ULA. Assume that the number of antennas possessed by each eavesdropper and the transmitting node are denoted by $N_e$ and $N$, respectively,  and each receiver has the same number of antennas as its transmitting node.

For ULA configuration with $q$ antennas, the elements are placed along the y-axis of the propagation plane with $\Delta \tau$ spacing. Hence,  the array steering and response vectors for the transmitting node and its receiver are written as~\cite{moraitis2007indoor}
{{\begin{align}\label{18_tr_vect}
&\hspace{-0.3 cm}\textbf{a}_t (\varphi,q) ={\left[
  {1,~}{{e^{ - j\frac{{2\pi }}{\omega }\Delta \tau \sin ({\varphi })}},}{ \ldots,~}{{e^{ - j\frac{{2\pi }}{ \omega }(q- 1)\Delta \tau \sin ({\varphi })}}}
 \right]^T}
\end{align}
 and
\begin{align}\label{19_tr_vect}
&\hspace{-0.3 cm}\textbf{a}_r (\xi,q) ={\left[
  {1,~}{{e^{ - j\frac{{2\pi }}{\omega }\Delta \tau \sin ({\xi })}},}{ \ldots,~}{{e^{ - j\frac{{2\pi }}{\omega }(q - 1)\Delta \tau \sin ({\xi})}}}
 \right]^T},
\end{align}}}
respectively, where $\omega$ is the wavelength, $\varphi \sim U(0, 2 \pi)$ and $\xi \sim U(0, 2 \pi)$ are the azimuth angle of departure (AoD) and angle of arrival (AoA), respectively,  and $\left(\cdot\right)^T$ denotes transpose. The channel model is established as $\textbf{H}=\sqrt{L(R)}{\bf{ A}}\left(\xi_r,\varphi_t\right)$ with the ULA steering matrix ${\bf{ A}}\left(\xi_r,\varphi_t\right)=\textbf{a}_r(\xi_r,q) \textbf{a}_t^H (\varphi_t,q)$, {where $\left(\cdot\right)^H$ is the conjugate transpose.}

\begin{figure*}[!t]
\normalsize
\setcounter{equation}{25}\begin{align}\label{P_1_ULA}
\mathcal{P}_1^{\mathrm{ULA}}\left( {x} \right)=\exp \left\{   - 2\pi {\lambda _e}  \int_0^\infty \int_0^{2\pi} {\rm{\mathbf{1}}}\left(  {\max \{ {r_e},d\}} < \big(\frac{{P_t}{G_e(\varphi_{t_{e,o}})}\beta}{x \sigma_e^2}\big)^{\frac{1}{{\alpha _{\mathrm{LoS}}}}}\right)  \frac{{{f_{\Pr }}({r_e})}}{2\pi} {r_e} d{\varphi_{t_{e,o}}} d{r_e}  \right\}
\end{align}
\hrulefill 
\begin{align}\label{P_2_ULA}
\mathcal{P}_2^{\mathrm{ULA}}\left( {x} \right)=\exp \left\{   - 2\pi {\lambda _e} \int_0^\infty   \int_0^{2\pi} {\rm{\mathbf{1}}}\left(  {\max \{ {r_e},d\}} < \big(\frac{{P_t}{G_e(\varphi_{t_{e,o}})}\beta}{x \sigma_e^2}\big)^{\frac{1}{{\alpha _{\mathrm{NLoS}}}}}\right) \frac{1-{{f_{\Pr }}({r_e})}}{2\pi} {r_e} d{\varphi_{t_{e,o}}}d{r_e} \right\}
\end{align}
\hrulefill 
\end{figure*}

We consider that  matched filter (MF) beamforming is adopted at all the nodes including eavesdroppers, the transmitting nodes and their receivers for maximizing the received signal power. Note that MF is the optimal beamforming for eavesdroppers, since interference is negligible at the eavesdroppers. Hence, the antenna gain for a typical transmitting node seen by its receiver is
\setcounter{equation}{18}\begin{align}\label{eq:G_o}
\hspace{-0.4 cm} G_o&= \left| \frac{\textbf{a}_{r}^H (\xi_{r_o},N)}{\sqrt{N}} {\bf{ A}}\left(\xi_{r_o},\varphi_{t_o}\right)  \frac{\textbf{a}_{t}
(\varphi_{t_o},N)}{\sqrt{N}}  \right|^2 =N^2,
\end{align}
{{ and the antenna gain for an interferer $i$ seen by the typical receiver is
\begin{align}\label{eq:G_i_1_1}
& G_i= \left| \frac{\textbf{a}_{r}^H (\xi_{r_o},N)}{\sqrt{N}}  {\bf{ A}}\left(\xi_{r_{i,o}},\varphi_{t_{i,o}}\right)  \frac{\textbf{a}_{t}
(\varphi_{t_i},N)}{\sqrt{N}} \right|^2.
\end{align}
Based on \eqref{18_tr_vect} and \eqref{19_tr_vect}, after some manipulations, we have
\begin{align}\label{eq:G_i}
G_i= \frac{1}{N^2}\frac{{\big[1 - \cos (N {\mathcal K}_1(\xi_{r_{i,o}}))\big] \big[1 - \cos ({N}{{\mathcal K}_2}(\varphi_{t_{i,o}},{\varphi _{{t_i}}}))\big]}}{{\big[1 - \cos ({\mathcal K}_1(\xi_{r_{i,o}}))\big] \big[1 - \cos ({{\mathcal K}_2}(\varphi_{t_{i,o}},{\varphi _{{t_i}}}))\big]}},
\end{align}}}
where $ {\mathcal K}_1\left(\xi_{r_{i,o}}\right)= 2\pi \frac{{\Delta \tau }}{\omega }(\sin({\xi_{{r_o}}}) - \sin (\xi_{r_{i,o}})) $, $ {{\mathcal K}_2} \left(\varphi_{t_{i,o}},{\varphi _{{t_i}}}\right)= 2\pi \frac{{\Delta \tau }}{\omega }(\sin(\varphi_{t_{i,o}}) - \sin ({\varphi _{{t_i}}}))$.

Based on \textbf{Theorem 2}, the average rate between the typical transmitting node and its intended receiver is lower bounded as
\begin{align}\label{ULA_LB_rate}
{\overline{R}}^{\mathrm{L}}_\mathrm{ULA}= {\log _2}\left( 1 + \frac{{ N^2 \beta }r^{-{\overline{\alpha}}}}{{ { \lambda \bar{G} \Lambda_\mathrm{ULA }+\frac {N_o}{P_t} } }}  \right),
\end{align}
where $\Lambda_\mathrm{ULA }$ is given from \eqref{Lambda_1_1_1} with the average antenna gain
{{ \begin{align}\label{G_average_ULA_1_1}
&\bar{G}=\mathbb{E}\left[G_i\right]=\frac{1}{N^2} \mathbb{E}\left[ \frac{{1 - \cos (N {\mathcal K}_1(\xi_{r_{i,o}})) }}{{1 - \cos ({\mathcal K}_1(\xi_{r_{i,o}})) }}  \right] \times \nonumber\\
 & \mathbb{E}\left[\frac{1 - \cos ({N}{{\mathcal K}_2}(\varphi_{t_{i,o}},{\varphi _{{t_i}}}))} {1 - \cos ({{\mathcal K}_2}(\varphi_{t_{i,o}},{\varphi _{{t_i}}})) } \right].
\end{align}
Since the beam-direction of the typical node and each interferer is a uniform random variable on $\left[0,2\pi\right]$, we can further obtain \begin{align}\label{G_average_ULA}
&\bar{G}=\frac{1}{N^2} \int_0^{2\pi} \frac{{1 - \cos (N {\mathcal K}_1(\xi_{r_{i,o}})) }}{{1 - \cos ({\mathcal K}_1(\xi_{r_{i,o}})) }} \frac{1}{2\pi}d{\xi_{r_{i,o}}} \times \nonumber\\
& \int_0^{2\pi} \int_0^{2\pi} \frac{1 - \cos ({N}{{\mathcal K}_2}(\varphi_{t_{i,o}},{\varphi _{{t_i}}}))} {1 - \cos ({{\mathcal K}_2}(\varphi_{t_{i,o}},{\varphi _{{t_i}}})) } \frac{1}{4\pi^2}  d\varphi_{t_{i,o}}  d{\varphi _{{t_i}}}.
\end{align}}}

Likewise, the antenna gain $G_e$ seen from the eavesdropper $e \in {\Phi _e}$  is
\begin{align}\label{eve_gain_ULA}
G_e\left(\varphi_{t_{e,o}}\right)&= \left| \frac{\textbf{a}_{r}^H (\xi_{r_{e,o}},N_e)}{\sqrt{N}}  {\bf{ A}}\left(\xi_{r_{e,o}},\varphi_{t_{e,o}}\right)  \frac{\textbf{a}_{t}
(\varphi_{t_o},N)}{\sqrt{N}} \right|^2 \nonumber\\
&=\left(\frac{N_e}{N}\right)^2\frac{{1 - \cos (N{{\mathcal K}_3}(\varphi_{t_{e,o}}))}}{{1 - \cos ({{\mathcal K}_3}(\varphi_{t_{e,o}}))}},
\end{align}
where $ {\mathcal K}_3\left(\varphi_{t_{e,o}}\right) =2\pi \frac{{\Delta \tau }}{\omega }(\sin(\varphi_{t_{e,o}}) - \sin ({\varphi _{{t_o}}}))$. 
 From \eqref{eve_gain_ULA}, we find that increasing the number of antennas at the transmitting node decreases the antenna gain obtained by the eavesdroppers, which is helpful for degrading the signal strength at the eavesdroppers. Based on \textbf{Theorem 3}, the exact average rate ${\overline{R}}_{e^{*}}^{\mathrm{ULA}}$ between the typical transmitting node and the most malicious eavesdropper is given from  \eqref{Eve_average_rate} by interchanging $\mathcal{P}_1\left( {x} \right) \rightarrow \mathcal{P}_1^{\mathrm{ULA}}\left( {x} \right)$ and $\mathcal{P}_2\left( {x} \right) \rightarrow \mathcal{P}_2^{\mathrm{ULA}}\left( {x} \right) $, where $\mathcal{P}_1^{\mathrm{ULA}}\left( {x} \right)$ and $\mathcal{P}_2^{\mathrm{ULA}}\left( {x} \right)$ are given by \eqref{P_1_ULA} and \eqref{P_2_ULA}, respectively. Thus, by using ULA, the average achievable secrecy rate can at least reach
\setcounter{equation}{27}\begin{align}\label{ULA_secrecy_rate}
{\overline{R}}_{s,\mathrm{ULA}}^L=\left[{\overline{R}}^{\mathrm{L}}_\mathrm{ULA}-{\overline{R}}_{e^{*}}^{\mathrm{ULA}}\right]^+.
\end{align}

\section{Artificial Noise  Aided Transmission}
In this section, we evaluate the secrecy performance for the artificial noise aided transmission~\cite{yang2015safeguarding}. For this case, the total power per transmission is $P_t =P_S+P_A$, where the power allocated to the information signal is $P_S=\mu P_t$,  and the power allocated to the artificial noise is $P_A=(1-\mu)P_t$. Here, $\mu$ is the fraction of power assigned to the information signal. The effective antenna gain $G_i^{S}$ for the information signal of an interfering $i$ seen by the typical receiver is expressed as
\begin{align}\label{G_S_array_gain}
G_i^{S}=
\left\{ \begin{array}{l}
  G_\mathrm{M}^{S} G_\mathrm{M},\quad {\Pr}_{\mathrm{M}\mathrm{M}}^{S}{\text{ = }} \frac{\vartheta \theta}{\left(2\pi\right)^2}  ,\hfill \\
  G_\mathrm{M}^{S} G_\mathrm{m},\quad {\Pr}_{\mathrm{M}\mathrm{m}}^{S}{\text{ = }} \frac{{\vartheta  \left( {2\pi  - \theta } \right)}}{{\left(2\pi\right)^2}} ,\hfill \\
  G_\mathrm{m}^{S} G_\mathrm{M},\quad {\Pr}_{\mathrm{m}\mathrm{M}}^{S}{\text{ = }} \frac{{ {(2\pi  - \vartheta)}\theta }}{{{\left(2\pi\right)^2}}}  ,\hfill \\
  G_\mathrm{m}^{S} G_\mathrm{m},\quad {\Pr}_{\mathrm{m}\mathrm{m}}^{S}{\text{ = }}\frac{{(2\pi  - \vartheta) (2\pi  - \theta) }}{\left({2\pi }\right)^2},
\end{array} \right.
\end{align}
where $\vartheta$, $G_\mathrm{M}^S$ and $G_\mathrm{m}^S$ are the beamwidth of the main-lobe, main-lobe gain and back-lobe gain for the information signal of an interfering $i$, respectively. Likewise,  {the} effective antenna gain for the artificial noise of an interfering $i$ seen by the typical receiver is expressed as
\begin{align}\label{G_AN_array_gain}
G_i^{A}=
\left\{ \begin{array}{l}
  G_\mathrm{M}^{A} G_\mathrm{M},\quad {\Pr}_{\mathrm{M}\mathrm{M}}^{A}{\text{ = }} \frac{\varsigma \theta}{\left(2\pi\right)^2}  ,\hfill \\
  G_\mathrm{M}^{A} G_\mathrm{m},\quad {\Pr}_{\mathrm{M}\mathrm{m}}^{A}{\text{ = }} \frac{{\varsigma  \left( {2\pi  - \theta } \right)}}{{\left(2\pi\right)^2}} ,\hfill \\
  G_\mathrm{m}^{A} G_\mathrm{M},\quad {\Pr}_{\mathrm{m}\mathrm{M}}^{A}{\text{ = }} \frac{{ {(2\pi  - \varsigma)}\theta }}{{{\left(2\pi\right)^2}}}  ,\hfill \\
  G_\mathrm{m}^{A} G_\mathrm{m},\quad {\Pr}_{\mathrm{m}\mathrm{m}}^{A}{\text{ = }}\frac{{(2\pi  - \varsigma) (2\pi  - \theta) }}{\left({2\pi }\right)^2},
\end{array} \right.
\end{align}
where $\varsigma$, $G_\mathrm{M}^A$ and $G_\mathrm{m}^A$ are the beamwidth of the main-lobe, main-lobe gain and back-lobe gain for the artificial noise of an interfering $i$, respectively. The effective antenna gain $G_e^{S}$ and $G_e^{A}$  for the information signal and artificial noise of the typical transmitting node seen by the eavesdropper $e \in {\Phi _e}$ can be respectively given from \eqref{G_S_array_gain} and \eqref{G_AN_array_gain} by interchanging the parameters $G_\mathrm{M}\rightarrow G_\mathrm{M}^e$, $G_\mathrm{m}\rightarrow G_\mathrm{m}^e$ and $\theta \rightarrow \phi$.

{Since the beam of the artificial noise at the typical transmitting node will not be directed to the typical receiver, the artificial noise sent by the typical transmitting node has negligible effect on the typical receiver~\cite{yang2015safeguarding},} the SINR at the typical receiver is given by
 \begin{align}\label{SINR_typical_receiver_AN}
{\widetilde{\gamma}_o} = \frac{{{P_S} G_\mathrm{M}^{S} G_\mathrm{M}  L\left(r\right) }}{{\sum\nolimits_{i \in \Phi /o} \left(P_S G_i^{S}+P_A G_i^{A}\right) L\left(\left|X_i\right|\right) + {\sigma_o^2}}}.
\end{align}
The SINR at the most malicious eavesdropper is given by
\begin{align}\label{SINR_eavesdropping_AN}
{\widetilde{\gamma}_{{e^*}}} = \mathop {\max }\limits_{e \in {\Phi _e}} \left\{ {\frac{{{P_S} G_e^{S} L\left(\left|X_e\right|\right)}}{{{{{P_A} G_e^{A} L\left(\left|X_e\right|\right)}+ \sigma_e^2}}}} \right\}.
 \end{align}
Following \eqref{average_Secrecy_rate}, the average achievable secrecy rate for the artificial noise aided transmission is written as
\begin{align}\label{SC_AN}
\widetilde{R}_S = {\left[\widetilde{R}-\widetilde{R}_e^{*}\right]^ + },
\end{align}
where $\widetilde{R}=\mathbb{E}\left[{{\log }_2}\left( {1 + {\widetilde{\gamma}_o}} \right)\right]$  and $\widetilde{R}_e^{*}=\mathbb{E}\left[ {{\log }_2}\left( {1 + {\widetilde{\gamma}_{e^{*}}}} \right) \right]$, $\widetilde{R}$ and $\widetilde{R}_e^{*}$ are given by the following theorems.
\begin{theorem}
 The exact average rate for the artificial noise aided transmission between the typical transmitting node and its intended receiver is given by
\begin{align}\label{average_rate_AN_VI}
\widetilde{R}=\frac{1}{\ln2}\int_0^\infty  {\frac{1}{{z}}(1 - \widetilde{\Xi}_1(z))\widetilde{\Xi}_2(z) {e^{ -z \sigma_0^2}}dz},
\end{align}
where $\widetilde{\Xi}_1(z)$ and $\widetilde{\Xi}_2(z)$ are respectively given by \eqref{Xi_1_AN} and \eqref{Xi_2_AN} at the top of next page. In \eqref{Xi_2_AN}, $\Pr_{\rm{M}}=\frac{\theta}{2 \pi}$ and $\Pr_{\rm{m}}=1-\Pr_{\rm{M}}$.
\begin{figure*}[!t]
\normalsize
\begin{align}\label{Xi_1_AN}
\widetilde{\Xi}_1(z)=f_\mathrm{Pr}\left(r\right)e^{- z{{P_S} G_\mathrm{M}^{S} G_\mathrm{M} \beta \left( {\max {{\{ r,d\} }}} \right)^{ - {\alpha_\mathrm{LoS}}}}}
+(1-f_\mathrm{Pr}\left(r\right))e^{- z{{P_S} G_\mathrm{M}^{S} G_\mathrm{M}\beta \left( {\max {{\{ r,d\} }}} \right)^{ - {\alpha_\mathrm{NLoS}}}}}
\end{align}
\hrulefill 
\begin{align}\label{Xi_2_AN}
\widetilde{\Xi}_2(z)=\exp\Big( - 2\pi \lambda \int_0^\infty  {f_{{\rm{Pr}}}}\left( u \right) ( 1 -\widetilde{\Omega}_1(z,u)) udu-2\pi \lambda \int_0^\infty (1-  {f_{{\rm{Pr}}}}\left( u \right)) ( 1 -\widetilde{\Omega}_2(z,u)) udu\Big)
\end{align}
with
\begin{equation*} 
\left\{\begin{aligned}
\widetilde{\Omega}_1(z,u)&=\sum\nolimits_{\ell, \nu,k  \in \left\{ {{\text{M}},{\text{m}}} \right\}}  \dfrac{{\Pr }_{\ell k }^{S} {\Pr }_{\nu k}^{A}}{\Pr_{k}}\times
{{{e}}^{ - z({P_S}{G_\ell^{S} G_k}+{P_A}{G_\nu^{A} G_k})\beta {\left( {\max {{\{ u,d\} }}} \right)^{ - {\alpha_{{\rm{LoS}}}}}}}}\\
\widetilde{\Omega}_2(z,u)&=\sum\nolimits_{\ell, \nu,k  \in \left\{ {{\text{M}},{\text{m}}} \right\}}  \dfrac{{\Pr }_{\ell k }^{S} {\Pr }_{\nu k}^{A}}{\Pr_{k}}\times
{{{e}}^{ - z({P_S}{G_\ell^{S} G_k}+{P_A}{G_\nu^{A} G_k})\beta {\left( {\max {{\{ u,d\} }}} \right)^{ - {\alpha_{{\rm{NLoS}}}}}}}}
\end{aligned}\right.
\end{equation*}
\hrulefill 
\end{figure*}
\begin{figure*}[!t]
\normalsize
\setcounter{equation}{40}\begin{align}\label{PP_3}
\widetilde{\mathcal{P}}_1\left( {x} \right)=&\exp \Big\{   - 2\pi {\lambda _e} \int_0^\infty {{f_{\Pr }}({r_e})}  {r_e}  \sum\nolimits_{\ell, \nu,n  \in \left\{ {{\text{M}},{\text{m}}} \right\}}  \dfrac{{\Pr }_{\ell n }^{S} {\Pr }_{\nu n}^{A}}{\Pr_{n}^e}{\rm{\mathbf{1}}}\left(  {\max \{ {r_e},d\}}< \big(\frac{{P_S}{G_\ell^S}{G_n^e}\beta-{P_A}{G_\nu^A}{G_n^e}\beta x}{x \sigma_e^2}\big)^{\frac{1}{{\alpha _{\mathrm{LoS}}}}}\right)  d{r_e} \Big\}
\end{align}
\hrulefill 
\begin{align}\label{PP_4}
\hspace{-0.3 cm}\widetilde{\mathcal{P}}_2\left( {x} \right)=&\exp \Big\{   - 2\pi {\lambda _e}  \int_0^\infty  (1-{{f_{\Pr }}({r_e})}) {r_e} \sum\nolimits_{\ell, \nu,n  \in \left\{ {{\text{M}},{\text{m}}} \right\}}  \dfrac{{\Pr }_{\ell n }^{S} {\Pr }_{\nu n}^{A}}{\Pr_{n}^e} {\rm{\mathbf{1}}}\left(   {\max \{ {r_e},d\}} < \big(\frac{{P_S}{G_\ell^S}{G_n^e}\beta-{P_A}{G_\nu^A}{G_n^e}\beta x}{x \sigma_e^2}\big)^{\frac{1}{{\alpha _{\mathrm{NLoS}}}}}\right) d{r_e} \Big\}
\end{align}
\hrulefill 
\end{figure*}
\end{theorem}
\begin{proof}
It can be proved by following a similar approach shown in the  {\textbf{Theorem} 1}.
\end{proof}

Using the similar approach shown in the Appendix B, the exact average rate given in \eqref{average_rate_AN_VI} can be lower bounded as a simple expression, which is given by the following theorem.
\begin{theorem}\label{theorem_4_1}
The lower bound of the average rate $\widetilde{R}$ is
\setcounter{equation}{36} \begin{align}\label{LB_average_rate_Alice_AN}
\widetilde{R}^{\mathrm{L}}_1= {\log _2}\left( 1 + \frac{{  G_{\text{M}}^S  G_{\text{M}} \beta }r^{-{\overline{\alpha}}}}{{ {\lambda  \widetilde{\Lambda}  + \frac{N_o}{\mu P_t} } }}  \right),
\end{align}
where  $\widetilde{\Lambda}$ is
\begin{align}\label{Lambda_AN_1_1}
&\widetilde{\Lambda} =\left( {{{\bar G}_S} +\frac{1-\mu}{\mu}{{\bar G}_A}} \right)\beta 2\pi   \nonumber \\
 & \times \Big( \int_{\text{0}}^d {({d^{ - {\alpha _{{\text{LoS}}}}}} - {d^{ - {\alpha _{{\text{NLoS}}}}}})} r{f_{{\text{Pr}}}}\left( r \right) + {d^{ - {\alpha _{{\text{NLoS}}}}}}rd{r}  \nonumber \\
 &\quad~~+\int_d^\infty  {({r^{1 - {\alpha _{{\text{LoS}}}}}} - {r^{1 - {\alpha _{{\text{NLoS}}}}}})} {f_{{\text{Pr}}}}\left( r \right) + {r^{1 - {\alpha _{{\text{NLoS}}}}}}dr \Big).
\end{align}
with
\begin{align*}
{{\bar G}_S } = \sum\nolimits_{\ell ,k \in \left\{ {{\text{M}},{\text{m}}} \right\}} {{G_{\ell }^S G_k}\text{Pr}_{\ell k}^S},
{{\bar G}_A } =  \sum\nolimits_{\nu ,k \in \left\{ {{\text{M}},{\text{m}}} \right\}} {{G_{\nu }^A G_k}\text{Pr}_{\nu k}^A}.
\end{align*}
\end{theorem}

Based on  {\textbf{Theorem} 5}, we have the following important corollary.
\begin{corollary}
The required average rate $\widetilde{R}_{\mathrm{th}}$ between the typical transmitting node and its receiver can be achieved when
the transmitting node density satisfies
\begin{align}\label{corollary_1}
\lambda \leq \left(\frac{{  G_{\text{M}}^S  G_{\text{M}} \beta }r^{-{\overline{\alpha}}}}{2^{\widetilde{R}_{\mathrm{th}}}-1}-\frac{N_o}{\mu P_t}\right) {\widetilde{\Lambda}}^{-1}.
\end{align}

\end{corollary}

We next present the average rate between the typical transmitting node and the most malicious eavesdropper as follows.
\begin{theorem}\label{theorem_4}
The exact average rate for the artificial noise aided transmission between the typical transmitting node and the most malicious eavesdropper is given by
\begin{align}\label{Eve_average_rate_AN}
\widetilde{R}_e^{*}= \frac{1}{{\ln 2}}\int_0^\infty  {\frac{{\left( {1 - \widetilde{\mathcal{P}}_1\left( {x} \right)\widetilde{\mathcal{P}}_2\left( {x} \right)} \right)}}{{1 + x}}dx},
\end{align}
where $\widetilde{\mathcal{P}}_1\left( {x} \right)$ and $\widetilde{\mathcal{P}}_2\left( {x} \right)$ are respectively given by \eqref{PP_3} and \eqref{PP_4}. In \eqref{PP_3} and \eqref{PP_4}, $\Pr_{\rm{M}}^e=\frac{\phi}{2 \pi}$ and $\Pr_{\rm{m}}^e=1-\Pr_{\rm{M}}^e$.

\end{theorem}
\begin{proof}
It can be proved by following a similar approach shown in the \textbf{Theorem } 2.
\end{proof}

Substituting \eqref{average_rate_AN_VI} and \eqref{Eve_average_rate_AN} into \eqref{SC_AN}, we obtain the average achievable secrecy rate for the artificial noise aided transmission.

\section{Numerical Results}
Numerical results are presented to understand the impact of mmWave channel characteristics and large antenna array on the achievable secrecy rate. We assume that the LoS probability function is $f_\mathrm{Pr}\left(R\right)=e^{-\varrho R}$ with $1/{\varrho}= 141.4$ m~\cite{Tianyang_arxiv2014}. The mmWave bandwidth is BW = 2 GHz, the noise figure is Nf = 10 dB, the noise power is $\sigma_o^2=\sigma_e^2 = -174 + 10\log 10$(BW)$ + \mathrm{Nf}$ dBm, and the reference distance is $d = 1$.

We focus on the carrier frequency at 28 GHz, 38 GHz, 60 GHz, and 73GHz, in which their LoS and NLoS path loss exponents are shown in  Table \ref{tab:passloss}  based on the practical channel measurements~\cite{deng201528,rappaport201238}.
\begin{table}\footnotesize
\centering
\caption{Path loss exponent for mm-wave outdoor channels \cite{deng201528,rappaport201238}.}\label{tab:passloss}
\begin{tabular}{|c|c|c|c|c|}
\hline
Path loss exponent & 28GHz & 38 GHz & 60 GHz & 73 GHz \\
\hline
LOS & 2 & 2  & 2.25  & 2\\
\hline
Strongest NLOS & 3  & 3.71 & 3.76 & 3.4\\
\hline
\end{tabular}
\end{table}
\begin{table}\footnotesize
\centering
\caption{Antenna Pattern \cite{venugopal2015interference}.}\label{tab:gain}
\begin{tabular}{|c|c|}
\hline
Number	of	antenna	 elements & $ N $ \\
\hline
Beamwidth $ \theta $ & $ \dfrac{2 \pi}{\sqrt{N}} $ \\
\hline
Main-lobe gain  & $ N $ \\
\hline
Side-lobe gain  & $ \dfrac{1}{\sin^2(3 \pi/2\sqrt{N})} $ \\
\hline
\end{tabular}
\end{table}

\subsection{Average Achievable Secrecy Rate}
In this subsection, we consider the uniform planar array (UPA) with the antenna pattern shown in Table~\ref{tab:gain}. The transmitting nodes and their receivers are equipped with $N$ antennas each, and each eavesdropper is equipped with $N_e$ antennas.

\begin{figure}\label{fig1}
\centering
\includegraphics[width=3.6in, height=3 in]{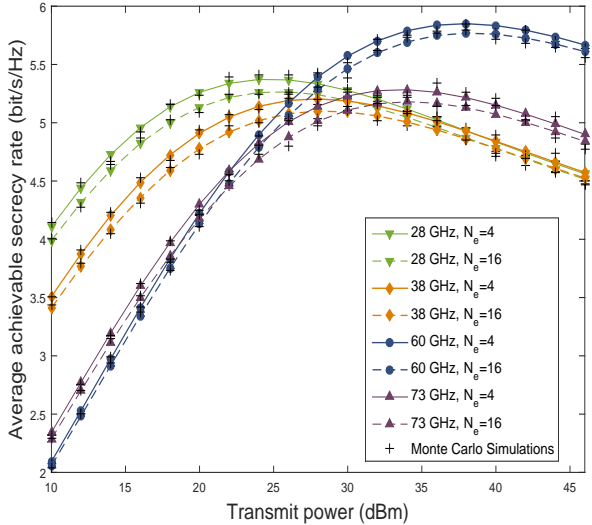}
\caption{ Effects of transmit power on the average achievable secrecy rate at 28 GHz, 38 GHz, 60 GHz and 73 GHz: $ \lambda=50/{\rm{km}}^2$, $ \lambda_e=100/{\rm{km}}^2$,  $N=16$, and  $r=15$ m.}
\end{figure}
{{Fig. 1 shows the effects of transmit power on the average achievable secrecy rate. We utilize four commonly-considered mmWave carrier frequencies, namely 28 GHz, 38 GHz, 60 GHz and 73 GHz, which have different $\beta$ values given by $\beta={(\frac{{\text{c}}}{{4\pi {f_c}}})^2}$ in Section II and path loss exponents in Table I. The analytical curves are obtained from  \eqref{average_Secrecy_rate}, which are validated by the Monte Carlo simulations marked by '$+$'.
We observe that there exist optimal transmit power values for maximizing average achievable secrecy rate at all the commonly-considered mmWave frequencies. In the low transmit power regime, better secrecy performance is achieved at 28 GHz, and higher average achievable secrecy rate can be obtained in the higher mmWave frequency band (60 GHz and 73 GHz) as the transmit power becomes large. The reason is that in the low transmit power regime, mmWave ad hoc network tends to be noise-limited, and mmWave link at lower mmWave frequencies experiences lower path loss and has stronger signal strength, which results in better { performance. However,} in the high transmit power regime, mmWave ad hoc network becomes  interference-limited. In this case, the interference received by a legitimate node becomes lower and the signal strength of the eavesdropper is also reduced at higher mmWave frequencies, due to the higher path loss at higher mmWave frequencies. In addition, it is shown that the secrecy performance at 60 GHz is better than that at 73 GHz when the transmit power is large enough, due to the fact that mmWave link at 60 GHz has higher LoS path loss exponent than that at 73 GHz~\cite{deng201528,rappaport201238} (2.25 at 60 GHz and 2 at 73 GHz in this figure based on the practical channel measurements in~\cite{deng201528,rappaport201238}), which leads to less interference received by a legitimate node and lower signal strength of the eavesdropper at 60 GHz.}}
Additionally, using the antenna pattern in Table~\ref{tab:gain}, average achievable secrecy rate is a bit lower at $N_e=16$ than that at $N_e=4$, due to fact that more effective antenna gain obtained by eavesdroppers using UPA with $N_e=16$, which deteriorates the secrecy performance.

\begin{figure}\label{fig:density_base_densityE}
\centering
\includegraphics[width=3.6in, height=3 in]{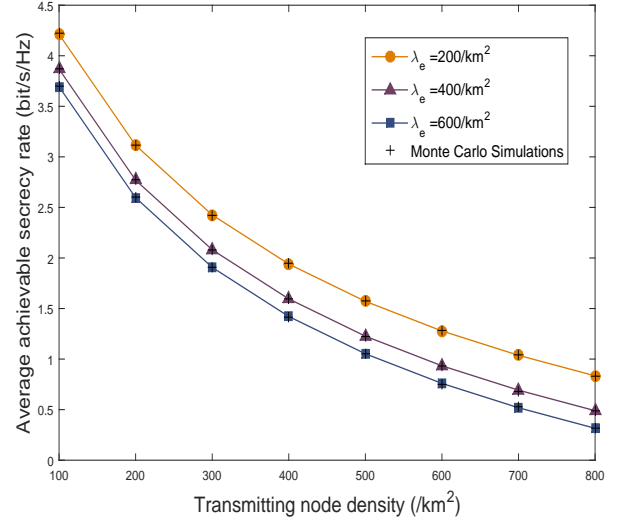}
\caption{Effects of transmitting node density on the average achievable secrecy rate at 60 GHz: $N=16$, $N_e=16$, $r=15$ m, and $P_t =30$ dBm.}
\end{figure}

Fig. 2 shows the effects of transmitting node density on the average achievable secrecy rate at 60 GHz. We see that when increasing the
transmitting node density, the average achievable secrecy rate declines. The reason is that when the transmitting nodes are
dense, mmWave ad hoc networks becomes interference-limited, and the interference caused by other transmitting nodes
dominate the performance. It is confirmed that in the large-scale mmWave ad hoc networks, more eavesdroppers
have a detrimental effect on the secrecy.
\begin{figure}\label{fig:N_base_Asymptotic}
\centering
\includegraphics[width=3.6in, height=3 in]{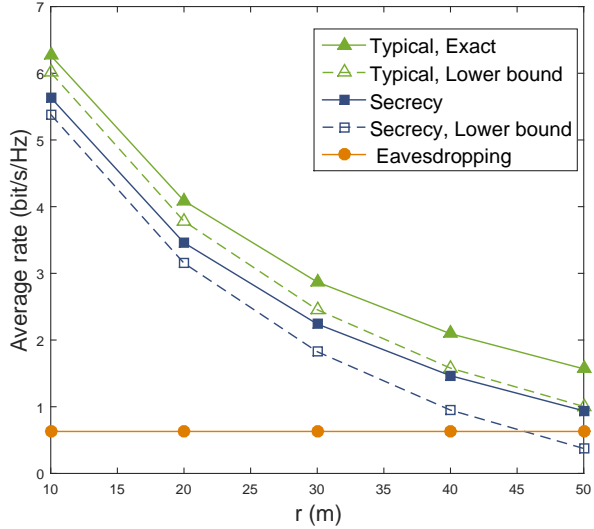}
\caption{Effects of transmit power with different typical distances on the average rate at 28 GHz: $P_t =10$ dBm, $\lambda=10/{\rm{km}}^2$, $\lambda_e=100/{\rm{km}}^2$, $N=16$,  and $N_e=16$.}
\end{figure}

Fig. 3 shows the effects of different typical distances on the average rate at 60 GHz. The green solid and dashed curves with triangles obtained from \eqref{R_average_rate} and \eqref{LB_average_rate_Alice} represent the exact and lower-bound average rate  between the typical transmitting node and its intended receiver, respectively,  and the orange solid curve with circles obtained from \eqref{Eve_average_rate} represents the average rate in the most malicious eavesdropping channel. We observe that the lower bound curves can efficiently predict the performance
behavior. It is shown that when the communication distance grows large, there is a significant decrease in the average achievable secrecy rate, due to the fact that the average rate between the typical transmitting node and its receiver decreases while the average rate in the most malicious eavesdropper's channel is unaltered. This illustrates that the secrecy rate in mmWave ad hoc networks is highly dependent on the communication distance between the transmitting node and its receiver.
\subsection{average achievable secrecy rate with ULA}
In this subsection, we consider the ULA configuration, and choose the antenna spacing as $ \bigtriangleup \tau= \frac{1}{2}\omega $. The results in Figs. 4 and 5 are obtained from \eqref{ULA_secrecy_rate}.

\begin{figure}\label{Fig:lowbound-ula}
\centering
\includegraphics[width=3.6in, height=3 in]{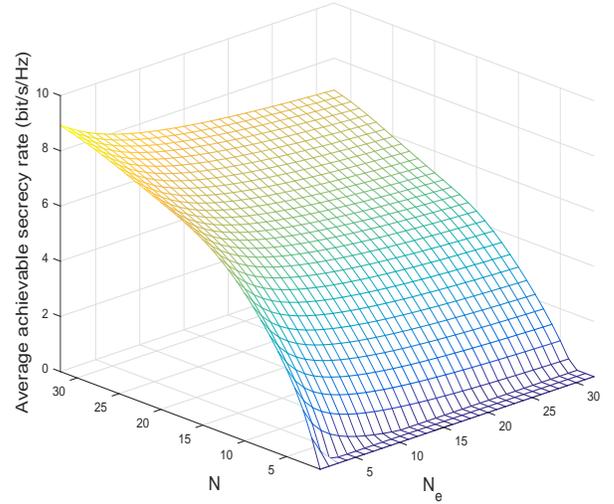}
\caption{ Effects of different antenna numbers on the average achievable secrecy rate at 38 GHz: $\lambda=50/{\rm{km}}^2$, $\lambda_e=100/{\rm{km}}^2$, $r=20$ m, $P_t =10 $ dBm, $ \xi_{r_o}=\pi/3 $, $ \varphi_{t_o}=\pi/3 $.}
\end{figure}

Fig. 4 shows the average achievable secrecy rate with different number of antennas at the transmitting nodes and eavesdroppers. It is observed that the average achievable secrecy rate increases with the number of antennas at the transmitting nodes, and decreases when eavesdroppers are equipped with more antennas. Moreover, the average achievable secrecy rate becomes very small when the transmitting node only has a couple of antennas. The reason is that the information signal beam is not narrow and more eavesdroppers can receive strong signals when they have more receive antennas.

\begin{figure}\label{Fig:lowbound-ula}
\centering
\includegraphics[width=3.6in, height=3 in]{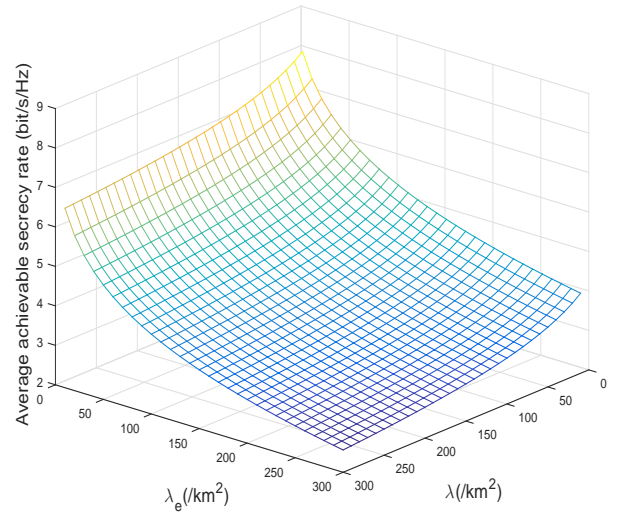}
\caption{ Effects of different node densities on the average achievable secrecy rate at 38 GHz: $N=16$, $N_e=4$, $r=20$ m, $P_t =10 $ dBm, $ \xi_{r_o}=\pi/3 $, $ \varphi_{t_o}=\pi/3 $.}
\end{figure}

Fig. 5 shows the achievable average achievable secrecy rate for different node densities.  We see that more eavesdroppers located in the networks are indeed harmful for secrecy. However, when the density of transmitting nodes increases, the secrecy performance also degrades, which indicates that interference can still be a concern for super dense transmitting nodes without highly directional antennas.

\subsection{average achievable secrecy rate with Artificial Noise}
In this subsection,  we examine the effects of artificial noise (AN) on the secrecy performance.
\begin{figure}\label{N_base_oneAN_Compare}
\centering
\includegraphics[width=3.6in, height=3 in]{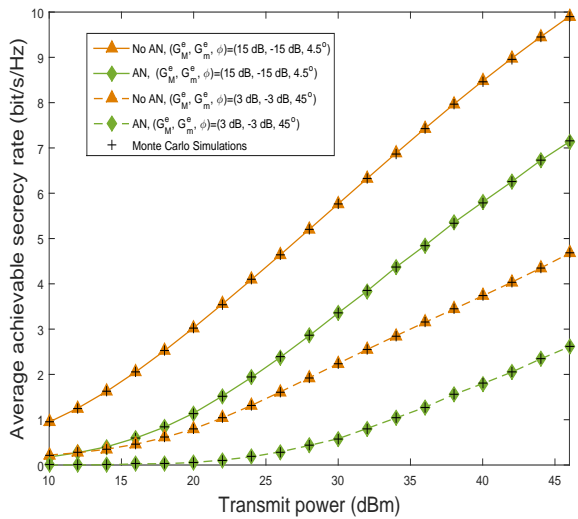}
\caption{ Effects of transmit power with/without AN on the average achievable secrecy rate at 60 GHz: $\lambda=20/{\rm{km}}^2$,  $\lambda_e=300/{\rm{km}}^2$,  $r=50$ m, and $\mu=0.85$.}
\end{figure}

Fig. 6 shows the effects of transmit power with/without AN at 60 GHz. We consider that the antenna beam patterns of sending information signal and AN at the transmitting node are $(G_\mathrm{M}^S, G_\mathrm{m}^S, \vartheta)=(3~\mathrm{dB}, -3~\mathrm{dB}, 45^{o})$ and $(G_\mathrm{M}^A, G_\mathrm{m}^A, \varsigma)=(3~\mathrm{dB}, -3~\mathrm{dB}, 45^{o})$, respectively, and the antenna beam pattern of only sending information signal without AN at the transmitting node is $(G_\mathrm{M}, G_\mathrm{m}, \theta)=(10~\mathrm{dB}, -10~\mathrm{dB}, 15^{o})$, as seen in~\cite{Andrew_Thornburg_2014}. The analytical curves without/with AN are obtained from \eqref{average_Secrecy_rate} and \eqref{SC_AN}, respectively. We see that when the transmitting nodes are not dense  ($\lambda=20/{\rm{km}}^2$ in this figure), the average achievable secrecy rate increases with the transmit power. {In this case, the use of AN with power allocation factor $\mu=0.85$ may not be able to improve secrecy\footnote{Note that the optimal power allocation for AN aided transmission is infeasible in the passive eavesdropping scenario, where the CSI of the eavesdropping channels cannot be obtained by the transmitter or legitimate receiver.}, and more power should be allocated to the information signal, to combat the severe interference and mmWave pathloss.  Such phenomenon has also been mentioned in the prior work ~\cite{Deng_TCOM_May_2016} with lower frequencies (See Fig. 7 in~\cite{Deng_TCOM_May_2016}), which is different from the results in the non large-scale physical layer security model.} Moreover, it is indicated that eavesdroppers using wide beam pattern can intercept more information.

\begin{figure}\label{Fig:SINR_LOW_BOUND}
\centering
\includegraphics[width=3.6in, height=3 in]{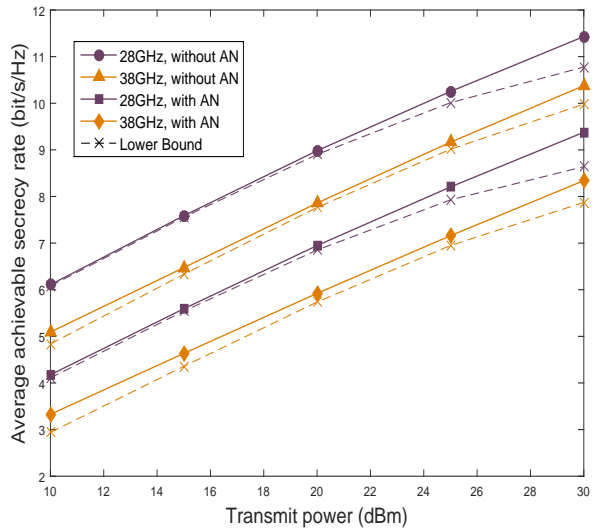}
\caption{ Effects of transmit power with AN on the average achievable secrecy rate at 28 and 38 GHz: $\lambda=30/{\rm{km}}^2$, $\lambda_e=500/{\rm{km}}^2$, $r =20$ m, $\mu=0.85$, $ (G_\mathrm{M}, G_\mathrm{m}, \theta)=(15~\mathrm{dB}, -15~\mathrm{dB}, 4.5^\circ)$,
 $ (G_\mathrm{M}^S, G_\mathrm{m}^S, \theta)=(10~\mathrm{dB}, -10~\mathrm{dB}, 15^\circ)$,
  $ (G_\mathrm{M}^A, G_\mathrm{m}^A, \theta)=(3~\mathrm{dB}, -3~\mathrm{dB}, 45^\circ)$,
     $ (G_\mathrm{M}^e, G_\mathrm{m}^e, \phi)=(3~\mathrm{dB}, -3~\mathrm{dB}, 45^\circ)$.}
\end{figure}

\begin{figure}\label{Fig:lowbound-ula}
\centering
\includegraphics[width=3.6in, height=3 in]{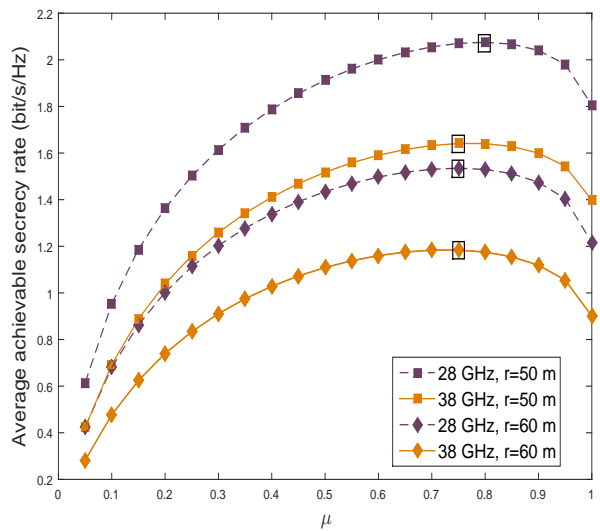}
\caption{ Effects of transmit power allocation factor on the average achievable secrecy rate at 28 and 38 GHz:  $\lambda=50/{\rm{km}}^2$, $\lambda_e=500/{\rm{km}}^2$, $P_t =30 $ dBm,
 $ (G_\text{M}, G_\text m, \theta)=(10~\mathrm{dB}, -10~\mathrm{dB}, 15^\circ)$,
   $ (G_\text M^S, G_\text m^S, \vartheta)=(3~\mathrm{dB}, -3~\mathrm{dB}, 45^\circ)$,
    $ (G_\text M^A, G_\text m^A, \varsigma)=(3~\mathrm{dB}, -3~\mathrm{dB}, 45^\circ)$.}
\end{figure}

Fig. 7 shows the effects of transmit power with/without AN in different frequency bands, i.e., 28 GHz and 38 GHz. The lower-bound results with/without AN  are obtained by using \eqref{LB_average_rate_Alice_AN} and \eqref{LB_average_rate_Alice} to calculate the average rate between the transmitting node and its receiver, respectively. We see that the lower bound results can well approximate the exact ones when the transmit power is not large ($<30$ dBm in this figure). The average achievable secrecy rate at 28 GHz is larger than that at 38 GHz, which indicates that the use of lower frequency bands could achieve better secrecy performance. The average achievable secrecy rate increases with transmit power, and the use of AN cannot improve the secrecy. The reason is that in this circumstance, more power should be used to enhance the transmission rate between the transmitting node and its receiver.

Fig. 8 shows the effects of transmit power allocation factor on the average achievable secrecy rate. We see that there exists an optimal $\mu$ to maximize the average achievable secrecy rate, which  reveals that AN can help enhance secrecy when the power allocation between the information signal and AN is properly set. Again, we see that larger communication distance $r$ deteriorates the secrecy performance. In addition, for a given $r$, secrecy transmission at 28 GHz is better than that at 38 GHz.




\section{Conclusion}
We concentrated on the secure communication in mmWave ad hoc networks by using physical layer security. We derived the average achievable secrecy rate without/with artificial noise. A tractable approach was developed to evaluate the average achievable secrecy rate when nodes are equipped with ULA. The results have highlighted the impacts of different mmWave frequencies, transmit power, node density and antenna gains on the secrecy performance. Important insights have been provided into the interplay between transmit power and mmWave frequency. When the node density is dense, the interference from nearby nodes dominates the secrecy performance. It is shown that power allocation between the information signal and AN  needs to  be carefully determined for secrecy performance enhancement.

{In this paper, we assume that the distance $r$ between the typical transmitting node and its receiver is constant. In the future work, we highlight that it is important to study the case of the dynamic $r$ following a certain distribution to model the specified scenarios. In addition, new antenna pattern models are needed to well characterize the effective antenna gain for a random interferer seen by the typical receiver when the number of mmWave antennas grows large.}

\section*{Appendix A: A detailed derivation of Theorem 1}
\label{App:theo_1}
\renewcommand{\theequation}{A.\arabic{equation}}
\setcounter{equation}{0}
Using \cite[Lemma 1]{Hamdi_2008}, the average rate ${\overline{R}}$ is calculated as
{\begin{align}\label{proof_1}
{\overline{R}}=&\mathbb{E} \left[  {{{\log }_2}\left( {1{\text{  +  }}{\gamma_0}} \right)} \right]
  {\text{ = }}\mathbb{E} \left[\frac{1}{\ln2}\int_0^\infty  {\frac{1}{{z}}(1 - {e^{ - z\gamma_o }}){e^{ - z}}dz}\right]\nonumber\\
  =&\frac{1}{\ln2}\mathbb{E}\left[ {\int_0^\infty  {\frac{1}{{z}}(1 - e^{- zY}){e^{ -z(\mathcal{I}+\sigma_0^2)}}dz} } \right] \nonumber\\
\mathop { = }\limits^{\left( a \right)}& \frac{1}{\ln2}\int_0^\infty  {\frac{1}{{z}}(1 -\underbrace{\mathbb{E}\left[e^{- zY}\right]}_{\Xi_1(z)})\underbrace{\mathbb{E}\left[{e^{ -z\mathcal{I}}}\right]}_{\Xi_2(z)} {e^{ -z\sigma_0^2}}dz},
\end{align}
where step (a) is obtained based on the fact that $Y$ and $\mathcal{I}$ are independent in the ad hoc networks,} $Y={P_t}G_{\text{M}}^2 L\left(r\right)$ is dependent on the LoS or NLoS condition  given a distance $r$, and the interference $\mathcal{I}$ is
\begin{align}\label{Int_exp}
\mathcal{I}= \sum\nolimits_{i \in \Phi/o} {{P_t}{G_i}L\left(\left|X_i\right|\right) }.
\end{align}
Based on the  {law} of total expectation, we can directly obtain $\Xi_1(z)$ as \eqref{Xi_z_x}.
Then, we see that $\Xi_2(z)$ is the Laplace transform of $\mathcal{I}$. To solve it, using the thinning theorem~\cite{Haenggi2009},  the mmWave  transmitting nodes are divided into two independent PPPs, namely LoS point process $\Phi_{\mathrm{LoS}}$ with density function $\lambda f_{\mathrm{Pr}}(R)$, and NLoS point process $\Phi_{\mathrm{NLoS}}$ with density function $\lambda(1-f_{\mathrm{Pr}}(R))$. Accordingly, by using the Slivnyak's theorem~\cite{Haenggi2009},  $\Xi_2(z)$ is given by
\begin{align}\label{I_laplace_exp}
\Xi_2(z)&=\mathbb{E}\left[e^{ -z\mathcal{I}}\right]=\mathbb{E}\left[{e^{ -z(\mathcal{I}_{\mathrm{LoS}}+\mathcal{I}_{\mathrm{NLoS}})}}\right] \nonumber\\
&=\mathbb{E}\left[{e^{ -z\mathcal{I}_{\mathrm{LoS}}}}\right] \mathbb{E}\left[{e^{ -z\mathcal{I}_{\mathrm{NLoS}}}}\right]
\end{align}
with
\begin{equation}\label{I_LoS_NLoS}
\left\{\begin{aligned}
\mathcal{I}_{\mathrm{LoS}} &= \sum\nolimits_{i \in \Phi_{\mathrm{LoS}}} {{P_t}{G_i}L\left(\left|X_i\right|\right)},\\
\mathcal{I}_{\mathrm{NLoS}} &=\sum\nolimits_{i \in \Phi_{\mathrm{NLoS}}} {{P_t}{G_i}L\left(\left|X_i\right|\right)}.
\end{aligned}\right.
\end{equation}
By applying the Laplace functional of   {the PPP~\cite{Haenggi2009},}
\begin{align}\label{I_LoS_Laplace_1}
&\mathbb{E}\left[{e^{ -z\mathcal{I}_{\mathrm{LoS}}}}\right]=\exp\Big( - 2\pi \lambda  \times  \nonumber\\
& \hspace{-0.4 cm} \int_0^\infty  {f_{{\rm{Pr}}}}\left( u \right)\Big( 1 -\underbrace{ \mathbb{E}\left[e^{ - z{P_t}{G_i}\beta \left( {\max {{\{ u,d\} }}} \right)^{ - \alpha _{\rm{LoS}} }}\right]}_{\Omega_1} \Big) udu \Big).
\end{align}
Based on the array gain distribution in \eqref{array_gain_pattern} and the law of total expectation, ${\Omega_1}$ is obtained as
\begin{align}\label{Omega1_array_gain}
\hspace{-0.4 cm}\Omega_1(z,u)=\sum\limits_{\ell, k  \in \left\{ {{\text{M}},{\text{m}}} \right\}} {\Pr }_{\ell k } \times
{{{e}}^{ - z{P_t}{G_\ell G_k}\beta {\left( {\max {{\{ u,d\} }}} \right)^{ - {\alpha_{{\text{LoS}}}}}}}}.
\end{align}
Likewise, we can derive $\mathbb{E}\left[{e^{ -z\mathcal{I}_{\mathrm{NLoS}}}}\right]$. Then, we get $\Xi_2(z) $ in \eqref{Xi_z_1}. Based on \eqref{proof_1} and \eqref{Xi_z_1}, we attain the desired result in \eqref{R_average_rate} and complete the proof.

\section*{Appendix B: A detailed derivation of Eq. \eqref{LB_average_rate_Alice}}
\label{App:theo_1}
\renewcommand{\theequation}{B.\arabic{equation}}
\setcounter{equation}{0}
{The average rate between the typical transmitting node and its intended receiver can be tightly lower bounded as~\cite{Lifeng_massiveMIMO}}
\begin{align}\label{proof1}
{{\bar{R}}_1^{\mathrm{L}}}& = {{\log }_2}\left( 1 + e^{\mathbb{E}\left[{\ln \gamma_o}\right]} \right),
\end{align}
where $\mathbb{E}\left[\ln\gamma _o\right]$ is calculated as
\begin{align}\label{E_ln_gamma}
\mathbb{E}\left[\ln\gamma _o\right]=&\underbrace{\mathbb{E}\left[\ln\left({{P_t}G_\mathrm{M}^2\beta {r^{ - {\alpha_o}}}}\right)\right]}_{Z_1}+\nonumber\\
&\underbrace{\mathbb{E}\left[\ln\left(\frac{1}{\sum\nolimits_{i \in \Phi /{o} } {{P_t}{G_i}\beta {{\left| {{X_{i,o}}} \right|}^{ - {\alpha _i}}}}  + {N_o}}\right)\right]}_{Z_2}.
\end{align}
Since the typical link can be either LoS or NLoS, using the law of total probability, $Z_1$ is calculated as
\begin{align}\label{Z_1_eq}
Z_1=\ln\left({P_t}G_\mathrm{M}^2\beta\right)-\left(f_\mathrm{Pr}\left(r\right) \alpha_{\mathrm{LoS}}+\left(1-f_\mathrm{Pr}\left(r\right)\right)\alpha_{\mathrm{NLoS}}\right) \ln r,
\end{align}
where $\alpha_{\mathrm{LoS}}$ and $\alpha_{\mathrm{NLoS}}$ are the path loss exponents of the LoS and the NLoS, respectively.

Considering the convexity of $\ln\left(\frac{1}{1+x}\right)$ and using Jensen's inequality, we derive the lower bound on the $Z_2$ as
\begin{align}\label{Z_2_step}
Z_2^{\rm{L}}=\ln\bigg(\frac{1}{\underbrace{\mathbb{E}\Big[\sum\nolimits_{i \in \Phi /{o} } {{P_t}{G_i}\beta {{\left| {{X_{i,o}}} \right|}^{ - {\alpha _i}}}} \Big]}_{\overline{\Lambda} }+N_o}\bigg).
\end{align}
Using a similar approach in \eqref{I_laplace_exp}, $\Lambda $ is derived as
\begin{align}\label{Int_step}
\begin{gathered}
  \overline{\Lambda}   =\mathbb{E}\Big[ {\sum\nolimits_{i \in {\Phi _{{\text{LoS}}}}} {{P_t}{G_i}\beta \left( {\max {{\{ \left| {{X_{i,o}}} \right|,d\} }^{ - {\alpha _{{\text{LoS}}}}}}} \right)} } \Big]\; \hfill \\
   + {\mathbb{E}}\Big[ {\sum\nolimits_{i \in {\Phi _{{\text{NLoS}}}}} {{P_t}{G_i}\beta \left( {\max {{\{ \left| {{X_{i,o}}} \right|,d\} }^{ - {\alpha _{{\text{NLoS}}}}}}} \right)} } \Big] \hfill \\
  \mathop { = }\limits^{\left( b \right)} {P_t}\bar G\beta 2\pi \lambda  \times \Big( \int_{\text{0}}^d \big({({d^{ - {\alpha _{{\text{LoS}}}}}} - {d^{ - {\alpha _{{\text{NLoS}}}}}})} r{f_{{\text{Pr}}}}\left( r \right) + {d^{ - {\alpha _{{\text{NLoS}}}}}}r \big) dr \hfill \\
  +  \int_d^\infty  \big({({r^{1 - {\alpha _{{\text{LoS}}}}}} - {r^{1 - {\alpha _{{\text{NLoS}}}}}})}
  {f_{{\text{Pr}}}}\left( r \right) + {r^{1 - {\alpha _{{\text{NLoS}}}}}} \big) dr\Big) ,
\end{gathered}
\end{align}
where $\bar{G}$ is the average array gain. Here, step (b) results from using Campbell's theorem~\cite{Baccelli2009}. Based on \eqref{array_gain_pattern} and using the law of total expectation, $\bar{G}$ is calculated as
\begin{align}\label{G_array_gain}
{{\bar G} } = {\mathbb{E}}\left\{ {{G_i  }} \right\}
= \sum\nolimits_{\ell ,k \in \left\{ {{\text{M}},{\text{m}}} \right\}} {{G_{\ell k}}\text{Pr}_{\ell  k}}.
\end{align}
Substituting \eqref{Z_1_eq}, \eqref{Z_2_step} and \eqref{Int_step} into \eqref{E_ln_gamma}, we obtain $\mathbb{E}\left\{\ln\gamma _o\right\}$ in  \eqref{proof1}, and the desired result \eqref{LB_average_rate_Alice}.

\section*{Appendix C: A detailed derivation of Theorem 2}
\label{App:theo_1}
\renewcommand{\theequation}{C.\arabic{equation}}
\setcounter{equation}{0}
The average rate ${\overline{R}}_{e^{*}}$ is calculated as
\begin{align}\label{theo2_11}
{\overline{R}}_{e^{*}}&=\mathbb{E} \left[ {{\log }_2}\left( {1 + {\gamma_{e^{*}}}} \right) \right]\nonumber\\
&=\frac{1}{{\ln 2}}\int_0^\infty  {\frac{{\left( {1 - {F_{{\gamma _{{e^*}}}}}\left( x \right)} \right)}}{{1 + x}}dx},
\end{align}
where ${F_{\gamma _{{e^*}}}}\left( \cdot\right)$ is the cumulative distribution function (CDF) of $\gamma _{{e^*}}$. By using the thinning theorem~\cite{Baccelli2009}, the eavesdroppers are divided into the  LoS point process $\Phi_e^{\mathrm{LoS}}$ with density function $\lambda_e f_{\mathrm{Pr}}(R)$, and NLoS point process $\Phi_e^{\mathrm{NLoS}}$ with density function $\lambda_e(1-f_{\mathrm{Pr}}(R))$. Then,  ${F_{\gamma _{{e^*}}}}\left( \cdot\right)$ is given by
\begin{align}\label{F_gamma}
{F_{{\gamma _{{e^*}}}}}\left( {x} \right)& = \Pr \left( {{\gamma _{{e^*}}} < {x}} \right) \nonumber\\
&=\Pr \left( \max\left\{\gamma _{{e^*}}^{\mathrm{LoS}},\gamma _{{e^*}}^{\mathrm{NLoS}}\right\} < {x}\right)\nonumber\\
&=\underbrace{\Pr \left( \gamma _{{e^*}}^{\mathrm{LoS}} < {x} \right)}_{\mathcal{P}_1\left( {x} \right)} \underbrace{\Pr \left( \gamma _{{e^*}}^{\mathrm{NLoS}} < {x} \right)}_{\mathcal{P}_2\left( {x} \right)},
\end{align}
where
\begin{equation}
\left\{\begin{aligned}
\gamma _{{e^*}}^{\mathrm{LoS}}&= \mathop {\max }\limits_{e \in {\Phi _e^{\mathrm{LoS}} }} \left\{ {\frac{{{P_t}{G_e}L\left(\left|X_e\right|\right) }}{{{\sigma_e^2}}}} \right\},\\
\gamma _{{e^*}}^{\mathrm{NLoS}}&= \mathop {\max }\limits_{e \in {\Phi _e^{\mathrm{NLoS}} }} \left\{ {\frac{{{P_t}{G_e}L\left(\left|X_e\right|\right) }}{{{\sigma_e^2}}}} \right\}.
\end{aligned}\right.
\end{equation}
We first derive $\mathcal{P}_1\left( {x} \right)$ as
\begin{align}\label{theo_eve_123}
& \mathcal{P}_1\left( {x} \right)=\Pr \left( \gamma _{{e^*}}^{\mathrm{LoS}} < {x} \right) \nonumber\\
&={\mathbb{ E}}\left[ {\prod \limits_{e \in \Phi _e^{\mathrm{LoS}}} {\Pr \left( {\frac{{{P_t}{G_e}\beta \left( {\max {{\{ {r_e},d\} }}} \right)^{ - {\alpha_\mathrm{LoS}}}}}{{{\sigma_e^2}}}} < {x} \right)} } \right] \nonumber\\
&\mathop {\rm{ = }}\limits^{\left( c\right)} \exp\Big\{-2\pi \lambda_e \times \nonumber\\
&\int_0^\infty \underbrace{\Pr \big( {\frac{{{P_t}{G_e}\beta \left( {\max {{\{ {r_e},d\} }}} \right)^{ - {\alpha_\mathrm{LoS}}}}}{{{\sigma_e^2}}}} >{x} \big)}_{\Theta}  f_{\mathrm{Pr}}(r_e) r_e d{r_e} \Big\},
\end{align}
where step (c) is obtained by using the Laplace functional. Based on the law of total probability, $\Theta$ is calculated as
\begin{align}\label{pro_law_total}
\Theta=\sum\limits_{
  \ell,n  \in \left\{ {{\text{M}},{\text{m}}} \right\} }  {\rm{\mathbf{1}}}\left(  {\max \{ {r_e},d\}} < \big(\frac{{P_t}{G_\ell }{G_n^e}\beta}{x \sigma_e^2}\big)^{\frac{1}{{\alpha _{\mathrm{LoS}}}}}\right) {{{\Pr }_{\ell n}}},
\end{align}
Substituting \eqref{pro_law_total} into \eqref{theo_eve_123}, we get $\mathcal{P}_1\left( {x} \right)$ in \eqref{PP_1}. Then, $\mathcal{P}_2\left( {x} \right)$ is similarly derived as \eqref{PP_2}.

\bibliographystyle{IEEEtran}


\end{document}